\documentclass[11pt,letterpaper,onecolumn]{IEEEtran}
\usepackage{amsmath}
\usepackage{amssymb}
\usepackage{epsfig}
\usepackage{subfig}
\usepackage{verbatim}
\usepackage{paralist}
\usepackage{tabulary}

\DeclareSymbolFont{calsymbols}{OMS}{cmsy}{m}{n} \DeclareSymbolFontAlphabet{\mathcal}{calsymbols}

\def\E{\mathop\mathbb{E}\nolimits}
\def\G{\mathop\mathbb{G}\nolimits}

\begin{document}

\newtheorem{theorem}{Theorem} \newtheorem{lemma}{Lemma}
\newtheorem{definition}{Definition} \newtheorem{corollary}{Corollary}
\newtheorem{proposition}{Proposition}
\newtheorem{example}{Example}

\title{Energy-efficient Scheduling of Delay Constrained Traffic over Fading Channels}

\author{\authorblockN{Juyul Lee and Nihar Jindal}\\
\authorblockA{Department of Electrical and Computer Engineering\\
University of Minnesota\\
E-mail: \{juyul,nihar\}@umn.edu}}

\maketitle


\begin{abstract}

A delay-constrained scheduling problem for point-to-point
communication is considered: a packet of $B$ bits must be
transmitted by a hard deadline of $T$ slots over a time-varying
channel. The transmitter/scheduler 
must determine how many bits to
transmit, or equivalently how much energy to transmit with, during
each time slot based on the current channel quality and the number
of unserved bits, with the objective of minimizing expected total
energy. In order to focus on the fundamental scheduling problem, it
is assumed that no other packets are scheduled during this time
period and no outage is allowed. Assuming transmission at capacity
of the underlying Gaussian noise channel, a closed-form expression
for the optimal scheduling policy is obtained for the case $T=2$ via
dynamic programming; for $T>2$, the optimal policy can only be
numerically determined. Thus, the focus of the work is on derivation
of simple, near-optimal policies based on intuition from the $T=2$
solution and the structure of the general problem. The proposed
bit-allocation policies consist of a linear combination of a
delay-associated term and an opportunistic (channel-aware) term. In
addition, a variation of the problem in which the entire packet must
be transmitted in a single slot is studied, and a channel-threshold
policy is shown to be optimal.
\end{abstract}


\section{Introduction} \label{sec-intro}

A time-varying channel is a fundamental feature of wireless
communication.  In this context, opportunistic scheduling refers to
the idea of transmitting with more power/higher rate when the
channel quality is good and less power/lower rate when the channel
is in a poor state. While this strategy is efficient from the
perspective of long-term average rate, it is not necessarily
appropriate for delay-constrained traffic which requires guaranteed
short-term performance.

In this paper we consider the problem of transmitting a packet of
$B$ bits over $T$ time slots, where the channel fades independently
 from slot to slot and the transmitter has perfect \textit{causal}
 channel information (i.e., knowledge of the current channel, but not of the future channel).
During each slot, the transmitter (or scheduler hereafter)
determines how many bits to transmit based on the current channel
quality and the number of bits yet to be served. The scheduler must
balance the desire to be opportunistic,
  i.e., wait to serve many of the bits when the channel is in a good state,
  with the hard deadline.  We investigate the setting where there is
  a single packet to be transmitted (i.e., no other packets are
  scheduled during the $T$ slot delay horizon),  the packet must be transmitted by the
  deadline, and transmission occurs at
  capacity of the underlying Gaussian noise channel.  In this framework our objective is to design
  a scheduling policy that minimize the expected energy consumed.
This setup reasonably models delay-constrained applications such as
VoIP, where packets arrive regularly and each must be received
within a short delay window.  In such a setting perhaps the most
important design objective is to minimize the resources (in our
case, energy) needed to meet the delay requirements.  In the
cellular uplink, for example, an energy-minimizing policy would
extend the battery life of mobile terminals.

\subsection{Prior Work}
Delay constrained scheduling in wireless communication systems has
been actively studied in various network settings under different
traffic models and delay constraints
(see for example
\cite{Berry_IT02}\cite{Rajan_IT04}\cite{Collins_Allerton99}\cite{Prabhakar_INFOCOM01}\cite{UysalBiyikoglu_IT04}\cite{Neely_IT07}\cite{Chen_INFOCOM07}\cite{Chen_WiOpt07}
and references therein). In \cite{Berry_IT02}\cite{Rajan_IT04}\cite{Collins_Allerton99},
power/rate control policies that minimize \textit{average delay} are
studied for a fading channel with random packet arrivals. In
\cite{Prabhakar_INFOCOM01}\cite{UysalBiyikoglu_IT04}\cite{Neely_IT07}\cite{Chen_INFOCOM07}\cite{Chen_WiOpt07}
systems with random packet arrivals, hard delay constraints, and
general energy-rate relationships are studied, but the emphasis is
on ``offline'' algorithms in which the scheduler has non-causal
knowledge of the packet arrivals and the channel states; heuristic
variations of the optimal ``offline'' algorithms are also proposed for
the more challenging ``online'' (i.e., \textit{causal}) setting.

In this paper, we rather focus on the interplay between fading, hard
deadlines, and causal channel information by studying transmission
of only a single packet, and thus do not consider random arrivals.
Not only is this model more tractable, but it also more reasonably
models applications with deterministic packet arrivals, e.g., VoIP
or video streaming.  To emphasize our treatment of physical-layer
issues, we use the terms \emph{causal} and \emph{non-causal} rather
than online and offline to indicate whether the scheduler has
knowledge of future channel states. Recently, Fu et
al.~\cite{Fu_WC06} considered this problem (single packet
transmission over a block fading channel, subject to a hard
deadline) and formulated it as a finite-horizon dynamic program
(DP).  For general energy-bit functions this DP can only be solved
numerically, but in \cite{Fu_WC06} a closed-form description of the
optimal policy is derived for the special case where the energy-bit
relationship is linear and the channel state is restricted to be an
integer multiple of some constant.  In this work we specialize the
framework of \cite{Fu_WC06} to the case where the energy-bit
relationship is governed by the AWGN channel capacity formula, and
derive closed-form descriptions of the optimal policy for $T=2$ and
sub-optimal policies for $T > 2$. In \cite{Zafer_WITA07} the work of
\cite{Fu_WC06} is extended to a setting where the channel evolves
according to a continuous Markov process, and the optimal scheduler
is derived for the case where the energy-bit relationship is given
by the AWGN capacity formula under particular assumptions on the
channel model (channels with \textit{drift}). However, these results
do not apply to the block fading model considered here and the
policies are rather different in structure from those developed
here.

In an earlier work, Negi and Cioffi \cite{Negi_IT02} studied the
dual problem of maximizing the expected number of transmitted bits
in a finite number of slots subject to a finite energy constraint
(with the energy-bit relationship described by the AWGN capacity
formula).  The optimal policy can generally only be found by
numerical methods (although a threshold policy is found to be
optimal at low SNR), and thus the solutions give little insight into
how the scheduling parameters (e.g., channel state, number of bits
to serve, number of slots remaining toward the deadline, and the
like) affect the scheduling process. Although we deal primarily with
suboptimal scheduling policies, we are able to
 deduce the effect of these parameters on the optimal policy.

\subsection{Summary of Contribution}
In this paper, we develop low-complexity and near-optimal scheduling
policies for delay-constrained causal scheduling. Our main result is
the following scheduler: a time-dependent weighted sum of a delay
associated term and an opportunistic term as
\begin{equation} \label{eq-framework}
b_t = \underbrace{\frac{1}{t} \beta_t}_{\text{delay associated}} + \underbrace{\frac{t-1}{t} \log \frac{g_t}{\eta_t}}_{\text{opportunistic}},
\end{equation}
where $b_t$ is the number of bits to serve (from the remaining
$\beta_t$ bits) at time slot $t$ ($t$ is in descending order and
thus represents the number of remaining slots), $g_t$ denotes the
current channel state, and $\eta_t$ denotes a channel threshold
determined by the channel statistics and the particular policy. If
the current channel quality is equal to the threshold level, then a
fraction $\frac{1}{t}$ of the remaining
 bits are transmitted. If the channel quality is better/worse than the threshold, then additional/fewer bits are transmitted. The scheduler acts very
 opportunistically when the deadline is far away ($t$ large) but less so as the deadline approaches. The motivation of
 this form was raised from the simple $T=2$ case, for which this form is shown to be optimal.
 
Two different suboptimal policies in the form of
\eqref{eq-framework} are proposed, one through a simple extension of
the optimal $T=2$ scheduler and the other by solving a relaxed
version of the optimization.  Numerical results are presented to
illustrate that these policies provide a significant advantage over
a naive equal-bit policy, and that they perform quite close to the
optimal for moderate/large values of $B$. In addition, we consider
the case of one-shot allocation where the entire packet must be
transmitted in only one of the slots.  This is an optimal stopping
problem, from which it follows that a simple channel threshold
policy is optimal.

This paper is organized as follows. Section II describes the problem formulation.
Section III discusses the optimal scheduler and Section IV develops suboptimal schedulers and their general
framework that gives an insight on the algorithm structure that reveals the incorporation of the delay constraint
on the scheduling process. Section V provides analysis and simulations. Section VI considers the one-shot allocation problem. We conclude in Section VII.

\emph{Notations:} The operation $\E[X]$ for a random variable $X$ denotes the expected value.
The operation $\G[X]$ for a random variable $X$ denotes $e^{\E[\ln X]}$ and the function $\G(x_1,\cdots,x_m)$ for deterministic quantities $x_1,\cdots,x_m$ denotes the geometric mean $(\prod_{i=1}^m x_i)^{1/m}$. The operation $\langle \cdot \rangle_x^y$ denotes truncation from below at $x$ and truncation from above at $y$. The function $1_{\{\cdot\}}$ denotes the indicator function, i.e., its value is 1 if the argument is true and 0 otherwise. The sets $\mathbb{R}_+$ and $\mathbb{R}_{++}$ denote the set of non-negative numbers and the set of positive numbers, respectively.

    \section{Problem Formulation}

We consider a single-user delay constrained scheduling problem as
illustrated in Fig.~\ref{fig:scheduling_simple}: a packet of $B$
bits must be transmitted within $T$ time slots through a fading
channel, in which $T$ is referred to as the \emph{delay-limit} or
\emph{deadline}.  We assume no other packet is scheduled during the
$T$ time slots, and that the packet must be transmitted by the
deadline (i.e., no outage is allowed).  Although these two
assumptions may not be entirely realistic, even for relatively
deterministic traffic (e.g., in VoIP, the next packet generally
arrives before the deadline of the previous has expired;
furthermore, a small percentage of packets are allowed to miss their
deadlines), these set of assumptions allow for a relatively
tractable problem and allow us to focus on the central issue of
meeting deadlines based upon causal channel information. The purpose
of the scheduler is to determine the energy, or equivalently the
number of bits, to be served during each time slot such that the
expected energy is minimized and the bits are served by the deadline
$T$. 
\begin{figure}
\centering
\includegraphics[width=0.63\textwidth]{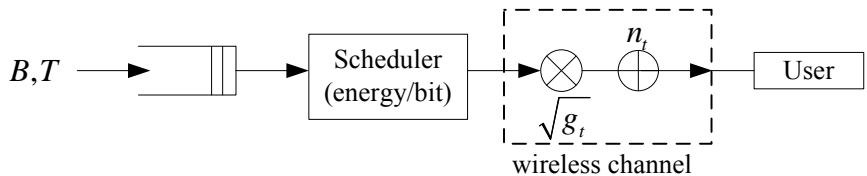}
\caption{Single-user delay constrained scheduling}
\label{fig:scheduling_simple}
\end{figure}

Time is indexed in descending order, i.e., $t=T$ is the initial
slot, $t=T-1$ is the 2nd slot, $\ldots$, and $t=1$ is the final slot
before the deadline; in doing so, $t$ represents the number of
remaining slots. The channel state, in power units, is denoted by
$g_t$. We assume that the channel states $\{g_t\}_{t=1}^T$ are
independently and identically distributed (i.i.d.) and the scheduler
has causal knowledge of these channel states (i.e., at time $t$,
$g_T,g_{T-1}\cdots,g_t$ are known but $g_{t-1},\cdots,g_1$ are
unknown). In this context, we refer to this type of scheduler as a
\emph{causal scheduler}.  The channel state $g$ is assumed to be a
non-degenerate positive continuous random variable.

Assuming unit variance Gaussian additive noise and transmission at
capacity, the number of transmitted bits, denoted as $b_t$, if $E_t$
energy is used is given by $b_t=\log_2(1+g_t E_t)$. By solving for
$E_t$ we arrive at a formula for the energy cost in terms of the
channel state $g_t$, and the number of bits\footnote{An implicit
assumption is that each slot spans $n$ channel symbols, for $n$
reasonably large, and that powerful coding allows for transmission
of $n b_t$ bits in the $t$-th slot. Thus, the quantity $b_t$ should
be thought of as the number of bits transmitted per channel symbol
during the $t$-th scheduling slot.} served $b_t$:
\begin{equation}
E_t (b_t,g_t) = \frac{2^{b_t}-1}{g_t}.
\end{equation}

We use $\beta_t$ to denote the queue state; i.e., the remaining bits
at time slot $t$. Then, $\beta_t$ can be calculated recursively as
$\beta_{t}=\beta_{t+1}-b_{t+1}$. Given this setup, a scheduler is a
sequence of functions $\{b_t\}_{t=1}^T$ that maps from the remaining
bits and the current channel state\footnote{Because the channel
states are assumed to be i.i.d., it is sufficient to make scheduling
decisions based only on the current channel (while ignoring past
channels).  If channels are correlated across time slots, then the
past and present channel should be used to compute the conditional
distributed of future channel states and all expected future energy
costs should be computed with respect to these conditional
distributions.} to the number of bits served, i.e.,
$b_t:\mathbb{R}_+ \times \mathbb{R}_{++}\to [0,\beta_t]$. Then, the
optimal energy-efficient scheduler is the set of scheduling
functions $\{b_t^\text{opt}(\cdot,\cdot)\}_{t=1}^T$ that minimizes
the total expected energy cost (summed over the $T$ slots): i.e.,
\begin{equation} \label{eq:E_sum_Et}
\min_{b_T, \cdots, b_1} \E\left[\sum_{t=1}^T E_t(b_t,g_t)\right]
\end{equation}
subject to $\sum_{t=1}^T b_t=B$ and $b_t\ge 0$ for all $t$.

The optimization in \eqref{eq:E_sum_Et} can be formulated
sequentially (via dynamic programming) with the remaining bits
$\beta_t$ as a state variable that summarizes the bit allocation up
until the previous time step.
\begin{equation} \label{eq:seq_Et}
b_t^\text{opt}(\beta_t, g_t)=\begin{cases}
\arg\min\limits_{0\le b_t\le \beta_t} \left\{E_t(b_t,g_t) + \E\left[\sum_{s=1}^{t-1} E_s (b_s,g_s)\Bigg\vert b_t\right]\right\},&t=T,\ldots,2,\\
\beta_1,& t=1.
\end{cases}
\end{equation}
This is the standard backward iteration: we first determine the
optimal action at $t=1$, then find the optimal policy at $t=2$ by
taking into account the optimal policy to be used at $t=1$, and so
forth.
Since $g_t$ is known but future channel states $g_{t-1},\ldots, g_1$ are unknown, the quantity $E_t$ is not random but the future energy costs $E_{t-1},\ldots,E_1$ are random.
Note also that the
optimization \eqref{eq:seq_Et} should be performed for all possible values of $\beta_t$ and $g_t$. In other
words, deriving the optimal scheduling function $b_t^\text{opt}$ is equivalent to finding the optimal
decision rule for all possible pairs $(\beta_t, g_t)$.

\section{Optimal Scheduling}
\label{sec:opt_causal}

In this section we attempt to derive the optimal (causal) scheduler
using the conventional dynamic programming technique
\cite{Bertsekas_DP1_Book05}. Unfortunately, an analytic expression
is obtained only when $T=2$ (besides the $T=1$ trivial case). For
$T>2$, we discuss the difficulty in obtaining an analytic
expression. When the scheduler has non-causal knowledge of the
future channel states, however, deriving an optimal scheduler is
possible; the optimal non-causal scheduler provides useful intuition
and is derived in Appendix \ref{sec:wf}.

\subsection{Optimal Scheduler for $T=2$} \label{sec:optimal_T2}
In the final time slot ($t=1$), the scheduler is required to transmit
all $\beta_1$ unserved bits regardless of the channel state $g_1$,
due to the hard delay constraint. Thus, the energy cost is given by
$E_1(\beta_1, g_1)=(2^{\beta_1}-1)/g_1$ for all $g_1$, and the
expected cost to serve $\beta_1$ bits in the final slot is
$\E_{g_1}\left[E_1(\beta_1,g_1)\right]=\E\left[\frac{1}{g}\right](2^{\beta_1}-1)$.

At $t=2$, $g_2$ is known but $g_1$ is unknown. The scheduler needs
to determine $b_2$, based on $g_2$ and $B$, while balancing the
current energy cost (of serving $b_2$ bits in the current slot) and
the \emph{expected} future cost (of deferring $B-b_2$ bits to the
last slot).  Thus, the optimum scheduler is the solution to the
following minimization:
\begin{eqnarray}
b_2^\text{opt}(B, g_2) &=& \arg
\min_{0\le b_2 \le B} \left(\underbrace{\frac{2^{b_2}-1}{g_2}}_{\text{current power cost}}+ \underbrace{\E_{ g_1}\left[E_1(B-b_2,g_1)\right]}_{\text{expected future cost}}
\right)\label{eq:J2_causal_T2}\\
&=& \arg\min_{0\le b_2\le B} \left( \frac{1}{g_2}\left(2^{b_2}-1\right)+\E\left[\frac{1}{ g_1}\right]\left(2^{B-b_2}-1\right)\right)\nonumber.
\end{eqnarray}
The objective function in \eqref{eq:J2_causal_T2} is convex, and
therefore the minimizer is found by setting the derivative to zero
while taking into account the constraints on $b_2$:
\begin{equation} \label{eq:b2_opt_T2}
b_2^\text{opt}(B, g_2) = \left\langle \frac{1}{2}B +
\frac{1}{2}\log_2 \left(g_2 \nu_1\right)\right\rangle_0^B,
\end{equation}
where $\nu_1 \triangleq \E\left[1/g\right]$ is a constant that
depends only on the distribution of the channel state $g$ (see
Appendix \ref{sec:nu} for the definition of constants $\nu_m$ for
$m=1, 2, \ldots$). Note that this policy depends only on the
unserved bits and the current channel state. This policy is only
meaningful when  $\nu_1$ is finite; this rules out Rayleigh fading,
in which case $g$ is exponentially distributed and thus
$\E\left[1/g\right]$ is not finite.

Notice that the optimal scheduling function \eqref{eq:b2_opt_T2} has two additive terms: (a) $\frac{1}{2}B$
corresponds to an equal distribution to time slots $t=1$ and $t=2$, and (b)
$\frac{1}{2}\log_2\left(g_2\nu_1\right)$ associated with a measure of the channel quality at
$t=2$. That is, if the channel quality $g_2$ is bigger than a threshold $1/\nu_1$, then more bits are
allocated than $\frac{1}{2}B$; if $g_t$ is smaller than the threshold then fewer bits are allocated and more bits are deferred to the final slot.

\subsection{Optimal Scheduler for $T>2$}

From \eqref{eq:seq_Et}, the optimization that the scheduler solves at each time step is:
\begin{equation} \label{eq:Jt_opt}
J_t^\text{opt}(\beta_t,g_t) =
\begin{cases}
\min\limits_{0\le b_t\le \beta_t} \left(\frac{2^{b_t}-1}{g_t}+\bar{J}_{t-1}^\text{opt}(\beta_t-b_t)\right), & t\ge 2\\
E_1(\beta_1,g_1), & t=1,
\end{cases}
\end{equation}
where
$\bar{J}_{t-1}^\text{opt}(\beta)=\E_{g}[J_{t-1}^\text{opt}(\beta,g)]$
denotes the \emph{cost-to-go} function, which is the expected cost
to serve $\beta$ bits in $(t-1)$ slots if the optimal control policy
is used at each step. This is a one-dimensional convex optimization
(pp.~87-88 in \cite{Boyd_Book04}) over $b_t$ and the optimal
solution satisfies
\begin{equation} \label{eq:bt_opt_det}
 b_t^\text{opt} (\beta, g_t)= \begin{cases}
  0, & g_t \le \frac{\ln 2}{(\bar{J}_{t-1}^\text{opt})'(\beta)},\\
  \arg_b \left\{\frac{2^{b}}{g_t} = \frac{1}{\ln 2}(\bar{J}_{t-1}^\text{opt})'(\beta-b)\right\}, & \frac{\ln 2}{(\bar{J}_{t-1}^\text{opt})'(\beta)} < g_t < \frac{2^{\beta}\ln 2}{(\bar{J}_{t-1}^\text{opt})'(0)}, \\
  \beta, & g_t \ge \frac{2^{\beta}\ln 2}{(\bar{J}_{t-1}^\text{opt})'(0)},
 \end{cases}
\end{equation}
assuming $\bar{J}_{t-1}^\text{opt}$ is differentiable (pp.~254-255
in \cite{Rockafellar_Book70}), where $\arg_b\{\cdot\}$ represents
the solution\footnote{Because of the convexity, the solution exists
uniquely if it exists.} of the argument equation.

When $t=2$, the cost-to-go function
$\bar{J}_1^\text{opt}(\beta)=(2^{\beta}-1)\nu_1$ (as well as its
derivative) takes on a very simple form  and thus
\eqref{eq:bt_opt_det} can be solved in closed form as in
\eqref{eq:b2_opt_T2}. However, the same is not true for $t>2$.
Because the optimal policy for $t=2$ is known, the cost-go-to
$\bar{J}_2^\text{opt}(\beta)$ can be written in closed form. The
derivative $(\bar{J}_2^\text{opt})'(\beta)$ can also be written in
closed form but cannot be analytically inverted; thus, the optimal
policy for $t=3$ can only be written in the form of
\eqref{eq:bt_opt_det} with the second condition given by the
following fixed point equation:
\begin{equation}
\frac{2^{b_3}}{g_3} = 2^{\beta-b_3}\int_0^{\frac{2^{-(\beta-b_3)}}{\nu_1}} \nu_1 dF(x) + 2^{\frac{\beta-b_3}{2}} \nu_1^{\frac{1}{2}} \int_{\frac{2^{-(\beta-b_3)}}{\nu_1}}^{\frac{2^{\beta-b_3}}{\nu_1}} \left(\frac{1}{x}\right)^{\frac{1}{2}}dF(x) + 2^{\beta-b_3}\int_{\frac{2^{\beta-b_3}}{\nu_1}}^\infty \frac{1}{x} dF(x),
\end{equation}
where $F$ is the cumulative distribution function of the channel
state $g$. As a result, no analytical characterization of
$\bar{J}_3^\text{opt}(\beta)$ is possible, and thus neither
$b_t^\text{opt}(\cdot,\cdot)$ nor $\bar{J}_t^\text{opt}(\beta)$ can
be found in closed form for $t\ge 4$.

Alternately, we can numerically find the optimal scheduler by the discretization method \cite{Bertsekas_AC75}. However, large complexity and memory is required for sufficiently fine discretization. More importantly, this numerical method gives little insight on how the delay constraint and channel state affect the scheduling function.

\section{Suboptimal Scheduling Policies}

Because the optimal scheduler cannot be written in closed form, it
is of interest to develop suboptimal schedulers. The first scheduler
is based on the intuition from the optimal $T=2$ policy, and the
second is found by solving a relaxed version of the optimization.

\subsection{Suboptimal I Scheduler}

If we compare the optimal causal scheduler for $T=2$ (Section \ref{sec:optimal_T2}) to the non-causal scheduler, we can immediately notice that the optimal scheduler determines $b_2^\text{opt}$ by inverse-waterfilling over channels $g_2$ and $1/\nu_1$, where the non-causal scheduler inverse waterfills over $g_2$ and the actual value of $g_1$\footnote{When both $g_2$ and $g_1$ are known at $t=2$, the optimal non-causal scheduling policy is given by $b_2^\text{IWF}(B,g_2)=\left\langle\frac{1}{2}B+\frac{1}{2}\log_2\left(\frac{g_2}{g_1}\right)\right\rangle_0^B$ from \eqref{eq:bt_wf_seq}, in which ``IWF'' stands for \emph{inverse waterfilling} (see Appendix \ref{sec:wf} for detail).}. This is because of the particularly simple form of the expected future cost. Although the expected future cost does not take on such a simple form for $T>2$, we can get a suboptimal scheduler by simply applying this inverse-waterfilling at every time slot $t$. In other words, at time
step $t$, perform inverse-waterfilling over the following $t$ channels:
\begin{equation*}
g_t, \underbrace{\frac{1}{\nu_1}, \ldots, \frac{1}{\nu_1} }_{t-1}
\end{equation*}
to determine how many of the unserved $\beta_t$ bits to serve now.
We denote this bit allocation policy as $b_t^{\rm (I)}$. Since $t-1$
of the $t$ channels are equal, the inverse-waterfilling operation is
very simple and the policy is given by
\begin{equation} \label{eq:bt_subopt1}
b_t^{\text{(I)}} (\beta_t, g_t) = \left\langle \frac{1}{t}\beta_t + \frac{t-1}{t}\log_2 \frac{g_t}{\eta_t^{\text{(I)}}} \right\rangle_0^{\beta_t},
\end{equation}
where $\eta_t^{\text{(I)}}=1/\nu_1$ serves as the channel threshold.
Notice that this threshold value depends only on the channel
statistics and is constant with respect to $t$.

When the deadline is far away (large $t$), the first term in
\eqref{eq:bt_subopt1} is negligible and the bit allocation is almost
completely dependent on the instantaneous channel quality. As the
deadline approaches ($t$ decreases toward $1$), the weight of the
channel-dependent second term decreases and the weight of the
delay-associated first term increases.

\subsection{Suboptimal II Scheduler}

The inability to find a general analytic solution to the original
optimization \eqref{eq:Jt_opt} is due to complications caused by the
constraint $0\le b_t\le \beta_t$ (for each $t$) in the dynamic
optimization. However, if we relax this constraint (i.e., allow $b_t
< 0$ and $b_t > \beta_t$ while maintaining the constraint
$\sum_{t=1}^T b_t = B$) we can derive the optimal policy in closed
form.

If we define the function $L_t$ as below, then we can show
inductively that $L_t$ represents the cost-to-go function for the
relaxed optimization:
\begin{equation} \label{eq:Lt}
L_t(\beta_t) = t 2^{\frac{\beta_t}{t}} \G(\nu_t,\nu_{t-1},\ldots,\nu_1)-t\nu_1
\end{equation}
where $\nu_1, \nu_2, \cdots$ are the fractional moments defined in
Appendix \ref{sec:nu} and $\G()$ represents the geometric mean
operation defined in Section \ref{sec-intro}. When $t=1$,
\eqref{eq:Lt} holds trivially. If we assume \eqref{eq:Lt} holds for
$t-1$, then the relaxed optimization for the next time step is given
by
\begin{equation} \label{eq:relaxed_opt}
\min_{b_t} \left(\frac{2^{b_t}-1}{g_t}+L_{t-1}(\beta_t-b_t)\right)
\end{equation}
and the solution (i.e., the optimum scheduler for the relaxed problem) is found by setting the derivative of the objective to zero:
\begin{equation} \label{eq:bt_subopt2_induct}
b_t = \frac{1}{t}\beta_t + \frac{t-1}{t}\log_2 \left(g_t \G(\nu_{t-1}, \ldots, \nu_1)\right).
\end{equation}
By plugging in the optimum value of $b_t$ in \eqref{eq:bt_subopt2_induct} into \eqref{eq:relaxed_opt} and taking expectation with respect to $g_t$, we reach \eqref{eq:Lt}.
By truncating the policy in \eqref{eq:bt_subopt2_induct} at $0$ and
$\beta_t$ we get a policy, referred to as Suboptimal II, for the
original (un-relaxed) problem:
\begin{equation} \label{eq:bt_subopt2}
b_t^{\rm (II)} = \left\langle\frac{1}{t}\beta_t + \frac{t-1}{t}\log_2 \frac{g_t}{\eta_t^{\rm (II)}}\right\rangle_0^{\beta_t},
\end{equation}
where
\begin{equation}
\eta_t^{\rm (II)} = \frac{1}{\G\left( \nu_{t-1}, \nu_{t-2}, \cdots, \nu_1\right)} \label{eq:eta_t_enhanced}
\end{equation}
denotes the threshold that depends only on the statistics not the realizations.

\subsection{Remarks on the Suboptimal Schedulers}

From \eqref{eq:bt_subopt1} and \eqref{eq:bt_subopt2}, we can see
that the two schedulers have a very similar form with the only
difference term being the threshold $\eta_t$.  Notice that both
policies simplify to the optimal policy for $t=2$. Based on the
policy formulations, this subsection investigates the common and
different characteristics of the suboptimal schedulers.

\subsubsection{General Framework}

The two algorithms thus far considered can be cast into a single framework:
\begin{equation} \label{eq:b_t_general}
b_t(\beta_t,g_t) = \left\langle \frac{1}{t}\beta_t + \frac{t-1}{t}\log_2 \frac{g_t}{\eta_t} \right\rangle_0^{\beta_t},
\end{equation}
where $\eta_t$ is the channel threshold determined by the individual
algorithms. This simple allocation strategy reveals how the delay
constraint affects the scheduling algorithms: at time step $t$ serve
a fraction $1/t$ of the remaining bits plus/minus a quantity that
depends on the strength of the current channel compared to a channel
threshold. If the current channel is good (i.e., $g_t$ is bigger
than the threshold $\eta_t$), additional bits are served (up to
$\beta_t$), while fewer bits are served when the current channel is
poorer than the threshold. Furthermore, note that when $t$ is large
(i.e., far from the deadline), the first term $\beta_t/t$ is very
small and the number of bits served is almost completely determined
by the current channel conditions. This agrees with intuition that
we should make aggressive, almost completely channel dependent (and
deadline independent) decisions when the deadline is far away, while
we should make more conservative (more deadline dependent, less
channel dependent) decisions near the deadline (small $t$).

Using $\log_2 10 \approx 3$ we can rewrite the policy in dB units
as:
\begin{equation}
b_t(\beta_t,g_t) \approx \left\langle \frac{1}{t}\beta_t + \left(\frac{t-1}{t}\right) \left(\frac{g_t^\text{dB}-\eta_t^\text{dB}}{3}\right)\right\rangle_0^{\beta_t}.
\end{equation}
For large $t$, approximately one bit is allocated for every $3$ dB
by which the channel exceeds the threshold.

\subsubsection{Channel Thresholds}
\label{sec:threshold}

The difference between the two policies is in the threshold values,
which are illustrated in Fig.~\ref{fig:threshold_subopt} for a
particular channel distribution. The suboptimal I scheduler has a
constant threshold $\eta_t^{\text{(I)}}=1/\nu_1$ for all $t$,
whereas Suboptimal II has a threshold that increases with $t$ (by
Proposition I). It is intuitive to use a larger threshold when the
deadline is far away (large $t$), as the scheduler can be more
selective because many different channels remain to be seen before
the deadline is reached.
\begin{figure}
\centering
\includegraphics[width=0.60\textwidth]{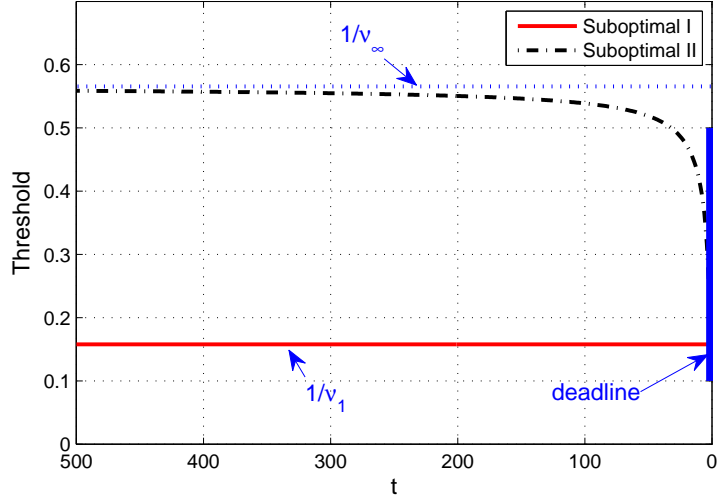}
\caption{Thresholds $\eta_t^\text{(I)}$ for the suboptimal I scheduler and $\eta_t^\text{(II)}$ for the suboptimal II scheduler when the channel state has the truncated exponential with $\gamma_0=0.001$.}
\label{fig:threshold_subopt}
\end{figure}

By using a constant threshold, Suboptimal I is not selective enough
and transmits too many bits when the deadline is far away. To see
this, consider the average number of bits transmitted in slot $t$
(ignoring truncation):
\begin{equation}
\E_{g_t}[b_t(\beta_t,g_t)] = \E_{g_t}\left[\frac{1}{t}\beta_t+\frac{t-1}{t}\log_2\frac{g_t}{\eta_t}\right]
= \frac{1}{t}\beta_t+\frac{t-1}{t}\E\left[\log_2\frac{g_t}{\eta_t}\right].
\end{equation}
Because $\eta_t^\text{(I)}=1/\nu_1=1/\E\left[1/g\right]$, by
Jensen's inequality
$\E\left[\log_2\frac{g_t}{\eta_t^\text{(I)}}\right] = \E\left[\log_2
g_t \right] +  \log_2 \E\left[1/g\right] > 0$.  Thus, Suboptimal I
transmits more than $\frac{B}{T}$ bits on average when scheduling begins, which is in some sense overly aggressive.
On the other hand, the quantity
$\E_{g_t}\left[\log_2\left(g_t/\eta_t^\text{(II)}\right)\right]$
decreases as $t$ increases and the limit is given by
\begin{equation}
\lim_{t\to\infty} \E_{g_t}\left[\log_2\frac{g_t}{\eta_t^\text{(II)}}\right] = 0
\end{equation}
because of Proposition \ref{prop:nu}. This implies that the
suboptimal II scheduler allocates $B/T$ bits on the average when the
deadline is far away and thus, unlike Suboptimal I, is not biased or
overly aggressive. Numerical results given later support the fact
that Suboptimal II generally performs better than Suboptimal I.

\subsection{Equal-bit Scheduler}

For comparison purposes, we consider one of the simplest causal schedulers: equal-bit scheduler. This policy allocates $B/T$ bits in each time slot, regardless of channel conditions, i.e.,
\begin{equation} \label{eq:bt_eq}
b_t^\text{eq}(\beta,g_t) = \frac{B}{T} = \frac{1}{t}\beta_t.
\end{equation}
The corresponding expected energy is given by
\begin{equation} \label{eq:EJt_eq}
\bar{J}_t^\text{eq}(\beta)= t(2^{\frac{\beta}{t}}-1)\E\left[\frac{1}{g}\right]=t(2^{\frac{\beta}{t}}-1)\nu_1.
\end{equation}

Although equal-power scheduling is asymptotically optimal for the dual problem of maximizing rate over $T$ slots when given a
finite energy budget in the high power regime \cite{Negi_IT02}, it will be seen that equal-bit scheduling is suboptimal
even when $B$ is large. 

\subsection{Inverse Waterfilling Interpretation}

If Suboptimal I and II and the equal-bit schedulers are compared to
the optimal non-causal policy (inverse waterfilling),  one can see
that each of the algorithms mimics inverse waterfilling using either
the current channel or channel statistics for the future channels,
as summarized in Table \ref{tbl:compare}.
\begin{table*}
\centering
\caption{Waterfilling analogy.} \label{tbl:compare}
\begin{tabular}{|l||c|}
\hline
& At each $t$, perform inverse-waterfilling over the following channels \\ \hline \hline
Equal-bit scheduler & $g_t, \underbrace{g_t,g_t, \cdots, g_t}_{t-1}$  \\ \hline
Suboptimal I scheduler & $g_t, \underbrace{\frac{1}{\nu_1},\frac{1}{\nu_1}, \cdots, \frac{1}{\nu_1}}_{t-1}$ \\ \hline
Suboptimal II scheduler & $g_t, \frac{1}{\nu_{t-1}}, \frac{1}{\nu_{t-2}}, \cdots, \frac{1}{\nu_1}$ \\ \hline
Non-causal IWF & $g_t, g_{t-1}, g_{t-2}, \cdots, g_1$ \\ \hline
\end{tabular}
\end{table*}

\section{Analysis \& Numerical Results}

In this section, we compare the performance of the optimal,
Suboptimal I and II, and equal-bit schedulers.  For $T=2$ we are
able to quantify the advantage of optimal scheduling relative to
equal bit scheduling in two extreme cases, while for $T>2$ we can
only consider numerical results.

\subsection{Asymptotic Analysis for $T=2$}

From the optimal scheduling expression for $T=2$ given in \eqref{eq:b2_opt_T2}, we can see that the
packet is split over both time slots (i.e., $0 < b_2 < B$) if and only if $2^{-B} /\nu_1 < g_2 < 2^{B}/\nu_1$. As $B\to 0$, the probability of this event goes to zero: if $g_2 < 1/\nu_1$ then all
bits are deferred to the final slot, while if $g_2 > 1/\nu_1$ all bits are served at $t=2$. As a result,
the expected energy cost takes on a rather simple form as $B\to 0$ (the derivation is provided in Appendix \ref{sec:pf_thm_ratio_T2}): 
\begin{equation}
\bar{J}_2^\text{opt}(B) \cong (2^B-1) \E\left[\min\left(\frac{1}{g_2},\nu_1\right)\right],\label{eq:EJ2_opt_smallB}
\end{equation}
where $\cong$ represents equivalence in the limit (i.e., the ratio
between both sides converges to $1$ as $B\to 0$). This implies that
the corresponding effective channel is $\max(g_2, 1/\nu_1)$. On the
other hand, when $B\to\infty$ the probability of only utilizing one
slot goes to zero and the limiting expected cost can be derived. The
following theorem quantifies the power advantage of optimal
scheduling:
\begin{theorem} \label{thm:ratio_T2}
The energy savings of optimal scheduling with respect to equal bit
scheduling in extremes of $B \rightarrow 0$ and $B \rightarrow
\infty$ is given by:
\begin{eqnarray}
\lim_{B\to 0}\frac{\bar{J}_2^{\text{eq}}(B)}{\bar{J}_2^{\text{opt}}(B)} &=&
\frac{\nu_1}{\E\left[\min\left(\frac{1}{g},\nu_1\right)\right]},   \label{eq:J_ratio_eq_opt_B_small}\\
\lim_{B\to\infty}\frac{\bar{J}_2^{\text{eq}}(B)}{\bar{J}_2^{\text{opt}}(B)}
&=& \sqrt{\frac{\nu_1}{\nu_2}} \label{eq:J_ratio_eq_opt_B_large}.
\end{eqnarray}
\end{theorem}
\begin{proof}
See Appendix \ref{sec:pf_thm_ratio_T2}.
\end{proof}

Table \ref{tbl:gaps} summarizes typical values of the energy savings
(at the extremes of $B\to 0$ and $B\to\infty$) for several fading
distributions, as given by Theorem \ref{thm:ratio_T2}. As
intuitively expected, the energy advantage is larger for more severe
fading distributions.  In other words, optimal scheduling is more
beneficial in more severe fading environments.
\begin{table*}
\centering
\caption{Average energy offsets for $T=2$} \label{tbl:gaps}
\begin{tabulary}{50pt}{|l||p{2.6cm}|p{2.6cm}|}
\hline
& \multicolumn{2}{c|}{equal-bit vs. optimal causal ($\bar{J}_2^\text{eq}(B)/\bar{J}_2^\text{opt}(B)$)}\\ \cline{2-3}
\raisebox{1.5ex}[0cm][0cm]{distribution of channel state $g$} & $B\to 0$ & $B\to\infty$ \\ \hline
truncated exponential with $\gamma_0=0.1$ & 1.96 dB & 0.44 dB  \\ \hline
truncated exponential with $\gamma_0=0.01$ & 3.26 dB & 1.04 dB \\ \hline
truncated exponential with $\gamma_0=0.001$ & 4.32 dB & 1.68 dB  \\ \hline
$1\times 2$ Rayleigh fading ($g\sim \chi^2_4$) & 1.99 dB & 0.52 dB \\ \hline
$1\times 3$ Rayleigh fading ($g\sim \chi^2_6$) & 1.37 dB & 0.27 dB   \\ \hline
$1\times 4$ Rayleigh fading ($g\sim \chi^2_8$) & 1.10 dB& 0.18 dB   \\ \hline
\end{tabulary}
\end{table*}

Figure \ref{fig:overall_avg_energy_T2} contains a plot of expected
energy versus $B$ for the optimal and equal-bit schedulers as well
as a plot of the energy difference between the two schedulers as a
function of $B$,  for channel state $g$ distributed as a truncated
exponential with the threshold $\gamma_0=0.001$.  The energy
advantage is seen to decrease from its $B \rightarrow 0$ advantage
of $4.32$ dB to the large $B$ asymptote of $1.68$ dB.
\begin{figure}
\centering
\subfloat[Average total energy]{\includegraphics[width=0.48\textwidth]{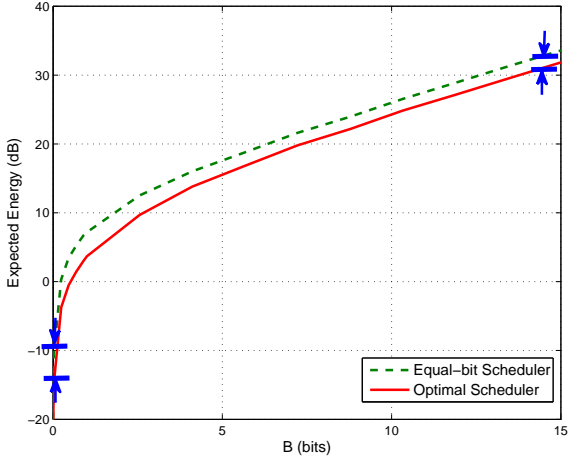}}\hfill
\subfloat[Energy advantage of optimal relative to equal-bit scheduling (difference in dB)]{\includegraphics[width=0.48\textwidth]{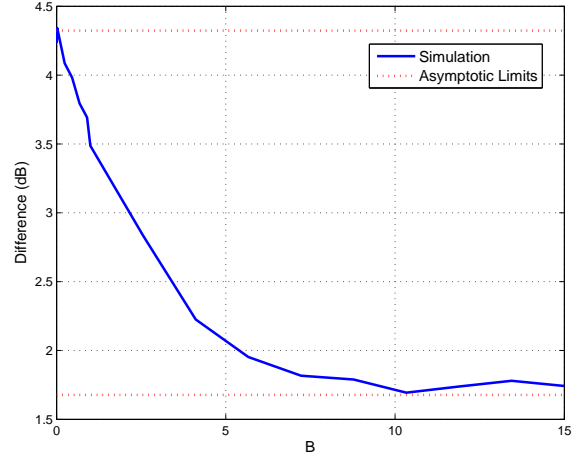}}
\caption{Average total energy consumptions for $T=2$ and average energy offset when $g$ is a truncated exponential variable with threshold $\gamma_0=0.001$ } \label{fig:overall_avg_energy_T2}
\end{figure}

\subsection{Numerical Results for $T>2$}

Throughout the simulations, we assume that the channel state $g_t$
is a truncated exponential with parameter $\lambda=1$ and threshold
$\gamma_0=0.001$. The factional moments of this truncated
exponential variable can be calculated as:
\begin{equation*}
\nu_m = \begin{cases}
\lambda e^{\lambda \gamma_0} \text{E}_1(\lambda \gamma_0), & m=1,\\
\lambda \left[e^{\lambda \gamma_0} \Gamma\left(\frac{m-1}{m}, \lambda\gamma_0\right)\right]^m, & m>1,
\end{cases}
\end{equation*}
where ${\rm E_1}(\cdot)$ and $\Gamma(\cdot,\cdot)$ denote the exponential integral and the incomplete gamma function, respectively, and its limit is given by $\nu_\infty =\frac{1}{\gamma_0} e^{-e^{\lambda\gamma_0}\text{E}_1(\lambda \gamma_0)}$.

Figures \ref{fig:overall_avg_energy_T5_T50}a and \ref{fig:overall_avg_energy_T5_T50}b compare the energy
consumption of the four different algorithms (equal-bit, Suboptimal I and II, optimal causal) for
$T=5$ and $T=50$, in which the optimal scheduler is calculated by numerical methods. The $x$-axis
denotes the total number of bits $B$ transmitted in $T$ time slots, and thus $B/T$ can be thought of as the
average bits per channel use. The $y$-axis denotes the average total energy cost $\bar{J}_T^\text{eq}$, $\bar{J}_T^\text{(I)}$, $\bar{J}_T^\text{(II)}$, and $\bar{J}_T^\text{opt}$.
\begin{figure}
\centering
\subfloat[$T=5$]{\includegraphics[width=0.48\textwidth]{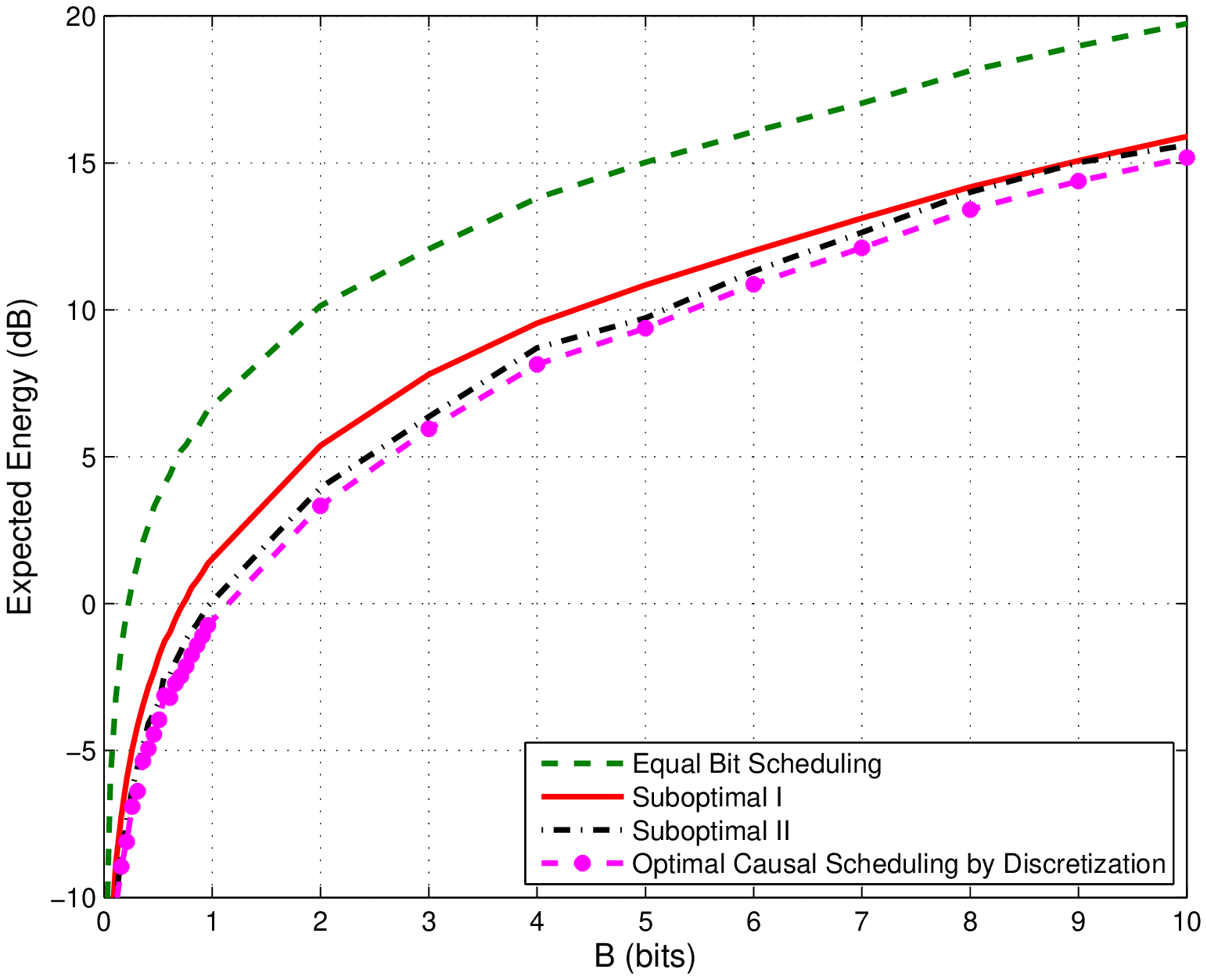}}\hfill
\subfloat[$T=50$]{\includegraphics[width=0.48\textwidth]{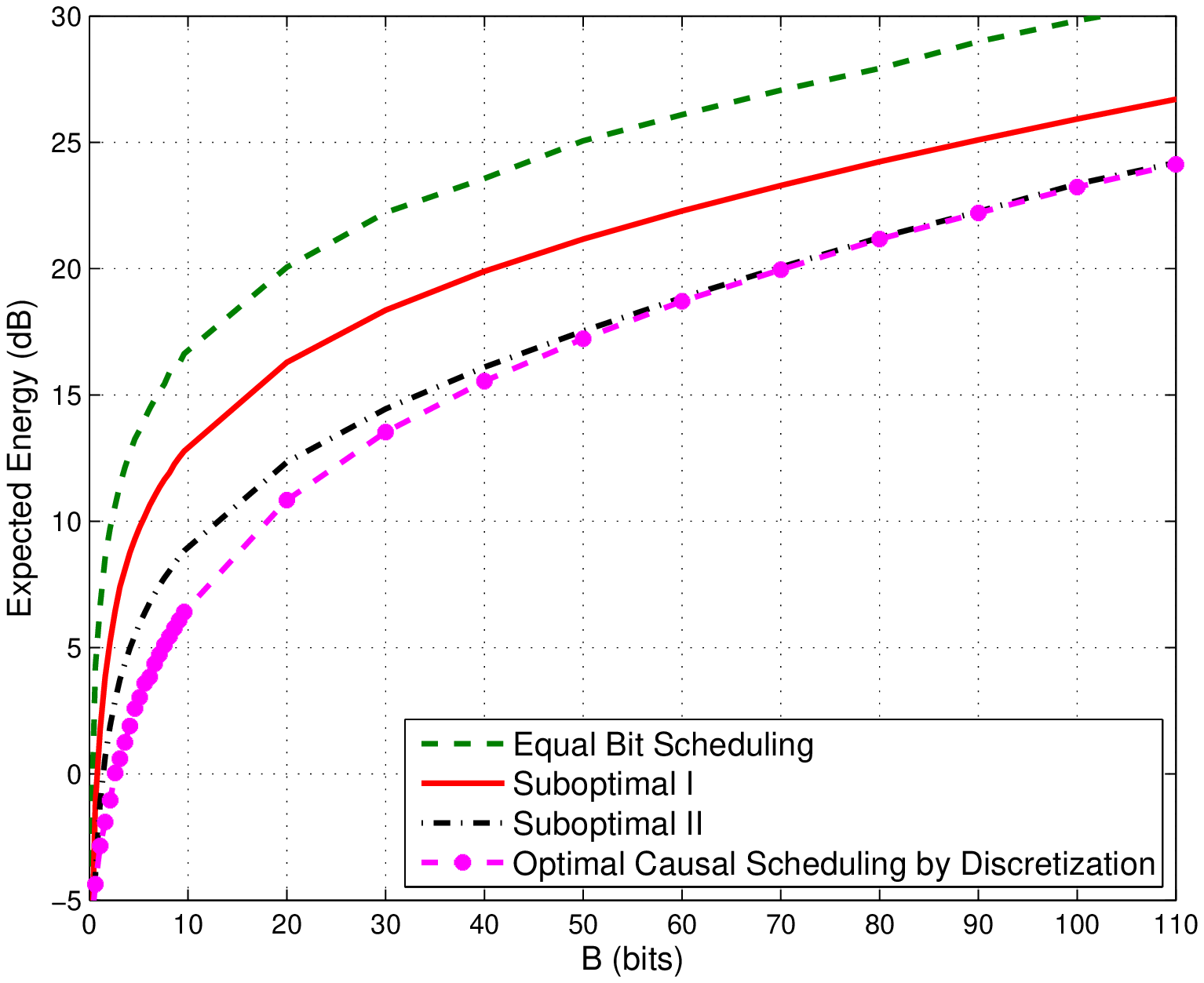}}
\caption{Average total energy consumption for $T=5$ and $T=50$} \label{fig:overall_avg_energy_T5_T50}
\end{figure}

From Fig.~\ref{fig:overall_avg_energy_T5_T50}a we see that both
Suboptimal I and II perform nearly as well as the optimal scheduler,
although Suboptimal II performs better than I. There are significant
differences between the equal-bit and optimal schedulers, which is
to be expected given the time diversity available over the five time
slots. In Fig.~\ref{fig:overall_avg_energy_T5_T50}b we see even
larger differences between equal-bit and optimal causal, which can
be explained by the even larger degree of time diversity ($T=50$).
Furthermore, Suboptimal II significantly outperforms Suboptimal I
for $T=50$ due to the over-aggressive nature of Suboptimal I.
Suboptimal II performs nearly as well as the optimal scheduler when
$B$ is approximately $50$ or larger (i.e., $B/T \ge 1$), but is
sub-optimal for smaller values
of $B$. 

Figure \ref{fig:avg_T10_profile} shows the expected bit allocation
$\E[b_t]$ for the different algorithms for $T=10$ slots when $B$ is
 large ($B=50$, upper) and small ($B=2$, lower). While the optimal
causal scheduling policy allocates roughly an equal number of bits
(averaged across different realizations, and not for each particular
realization) to each of the slots, Suboptimal I is immediately seen
to allocate too many bits (on average) to early time slots which
agrees with our earlier claim that this algorithm is often
overly-aggressive as explained in Section \ref{sec:threshold}. For
$B=50$ the bit allocation of Suboptimal II is very similar to that
of the optimal policy. However, for $B=2$ Suboptimal II is also
overly-aggressive as compared to the optimal. We suspect that the
performance of Suboptimal II could be further improved by performing
some heuristic modifications to the algorithm, but this is beyond
the scope of the paper and is left to future work.
\begin{figure}
\centering
\includegraphics[width=0.63\textwidth]{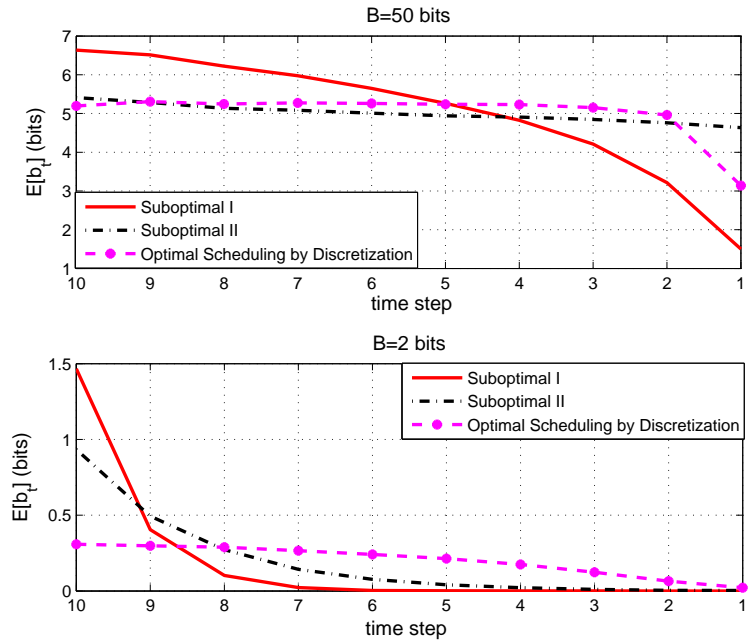}
\caption{Bit allocation profiles for $T=10$ when $B=50$ (upper) and $B=2$ (lower)} \label{fig:avg_T10_profile}
\end{figure}

To summarize, the numerical results indicate that (a) Suboptimal II
is nearly optimal for moderate to large values of $B$, (b)
Suboptimal II outperforms Suboptimal I, and (c) neither suboptimal
algorithm is near optimal for small values of $B$. In the next
section, we will consider a policy that performs close to the
optimal when $B$ is small.

\section{One-shot Allocation}

In some settings it may be undesirable to split the packet across
multiple time slots, e.g., because there is a large overhead
associated with each slot used for transmission. In this scenario we
may wish to find only one time slot among the
$T$ slots for the transmission of  $B$ bits; i.e., the action $b_t$ can be either $0$ or $B$. 

The dynamic program in this setting can be written as
\begin{eqnarray}
J_1(B) &=& \frac{2^{B}-1}{g_1},\\
J_{t}(B) &=& \min\left\{ \frac{2^{B}-1}{g_t}, \E[J_{t-1}(B)] \right\},\qquad t=2,\cdots,T, \label{eq:oneshot_dp}
\end{eqnarray}
which is precisely an optimal stopping problem
\cite{Bertsekas_DP1_Book05}. Thus, a threshold policy is optimal:
allocate all $B$ bits at the first slot $t$ such that $g_t >
1/\omega_t$, where $1/\omega_t$ is the threshold.  That is,
\begin{equation} \label{eq:stopping_bt}
b_t = \begin{cases}
B,& t=\max\left\{s: g_s>1/\omega_s\right\},\\
0,& \text{elsewhere}.
\end{cases}
\end{equation}
At $t=1$ a packet must be served and thus $\omega_1$ is infinite. Because the expected cost-to-go decreases as $t$ increases, the threshold also decreases with $t$. In Appendix \ref{sec:pf_stopping_omega_t} we show the thresholds are given by the following recursive formula.
\begin{equation} \label{eq:stopping_omega_t}
\omega_t = \begin{cases}
\E\left[\frac{1}{g}\right], & t=2,\\
\E\left[\frac{1}{g}\Big\vert \frac{1}{g_t}<\omega_{t-1}\right]\Pr\left\{\frac{1}{g_t}<\omega_{t-1}\right\}+\omega_{t-1}\Pr\left\{\frac{1}{g_t}\ge \omega_{t-1}\right\}, & t=3,\cdots, T.
\end{cases}
\end{equation}
Notice that the threshold $1/\omega_t$ depends only on the channel statistics and does not depend on $B$.
\begin{figure}
\centering
\includegraphics[width=0.6\textwidth]{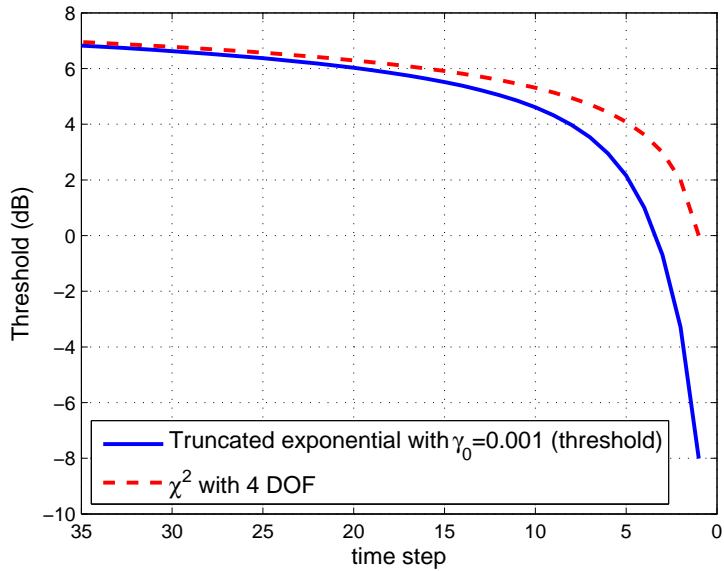}
\caption{Thresholds for the one-shot allocation} \label{fig:omega_update_together}
\end{figure}

Figure \ref{fig:omega_update_together} illustrates the thresholds
for the truncated exponential $g$ (with $\lambda=1$ and
$\gamma_0=0.001$) and the chi-squared $g$ (with 4 degrees of
freedom).  Figure \ref{fig:cmp_stopper} illustrates the energy usage
(normalized by $T$) of the optimal one-shot allocation policy and the multiple slot policies. The $B/T=0.1$ and $B/T=1$ curves
illustrate performance for relatively small and large values of $B$,
respectively.  When $B$ is small, the energy of the one-shot
allocation is nearly the same as the optimal policy that allows for
multiple slots to be used. However, this one-shot allocation is not
appropriate when $B$ is relatively large because the required energy
of the one-shot policy grows exponentially with $B$.
\begin{figure}
\centering
\includegraphics[width=0.63\textwidth]{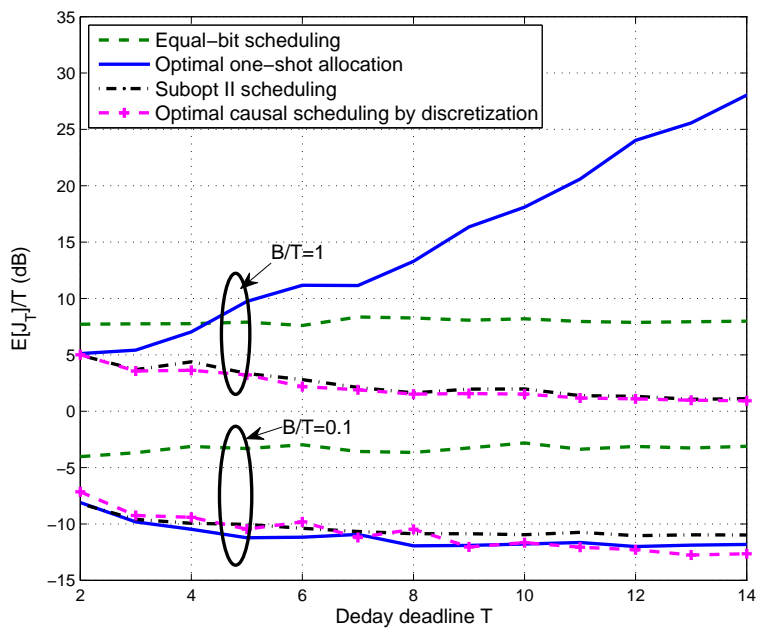}
\caption{Performance of the optimal one-shot allocation compared with multi-slot allocation algorithms} \label{fig:cmp_stopper}
\end{figure}

\section{Conclusion}

In this paper we considered the problem of bit/energy allocation for
transmission of a finite number of bits over a finite delay horizon,
assuming perfect instantaneous channel state information is
available to the transmitter and that the energy and rate are
related by the Shannon-type (exponential) function. We derived the
optimal scheduling policy when the deadline spans two time slots,
and derived two near-optimal policies for general deadlines. The
proposed schedulers have a simple and intuitive form that gives
insight into the optimal balance between channel-awareness (i.e.,
opportunism) and deadline-awareness in a delay-limited setting. We
also considered the same problem under the additional constraint
that only a single of the available time slots can be used, and in
this case found the optimal threshold-based policy.  Based upon the
policy constructions and the numerical results, we observed that the
suboptimal II scheduler is near-optimal for large/moderate values of
$B$ while the one-shot policy is near-optimal for small values of
$B$.

Given the increasing volume of delay-limited traffic over packet-switched wireless networks (e.g., VoIP or
multimedia transmission in 3G systems), we expect problems of this sort to become increasingly important. Of
course, the problem considered here represents only a particular instance of the rich space of delay-limited
scheduling problems. Interesting extensions include consideration of discrete code rates, peak power
constraints, and multi-user issues, and we hope this work provides useful insight for some of these other
formulations.

\appendices

\section{Non-causal Scheduling}
\label{sec:wf}

If the channel states are known non-causally, i.e., $g_T,\ldots, g_1$ are known at $t=T$, the optimal
scheduling/allocation is determined by waterfilling because each time slot serves as a
parallel channel. While conventional waterfilling maximizes rate subject to a power constraint,
this is the dual of minimizing power/energy subject to a rate/bit constraint and is referred to as \emph{inverse-waterfilling} (IWF):
\begin{equation} \label{eq:JT_IWF}
J_T^\text{IWF}(B,\{g_t\}_{t=1}^T)=\min_{b_T,\cdots,b_1} \sum_{t=1}^T \frac{2^{b_t}-1}{g_t},
\end{equation}
subject to $\sum_{t=1}^T b_t=B$ and $b_t\ge 0$.  This is a convex optimization problem and can be easily solved using the standard Lagrangian method:
\begin{equation} \label{eq:iwf_simple}
b_t^\text{IWF}=\left\langle \log_2 \left(\frac{g_t}{g_\text{th}}\right) \right\rangle_0^\infty,
\end{equation}
where $g_\text{th}$ is the solution to $\sum_{i=1}^T \left\langle \log_2\left(\frac{g_i}{g_\text{th}}\right) \right\rangle_0^\infty = B$.
A time slot $t$ is called \emph{utilized} if a positive bit is scheduled at $t$, i.e., $b_t>0$ or equivalently $g_t>g_\text{th}$.
With algebraic manipulations, we can express this IWF policy in \eqref{eq:iwf_simple} sequentially like other causal scheduling policies as
\begin{equation}
b_t^\text{IWF}(\beta_t,g_t) =
\frac{1}{t'}\beta_t+\frac{t'-1}{t'}\log_2 \frac{g_t}{\eta_t^\text{IWF}}, \qquad \text{if}\;\; g_t > g_\text{th},
\label{eq:bt_wf_seq}
\end{equation}
otherwise $b_t^\text{IWF}(\beta_t,g_t)=0$,
where $t'=\sum_{i=1}^t 1_{\{g_i\ge g_{\rm th}\}}$ and $\eta_t^\text{IWF}=\left(\prod\limits_{i=1}^{t-1} g_i^{1_{\{g_i>g_\text{th}\}}}\right)^{1/(t'-1)}$. Notice that $g_{t-1},\cdots,g_1$ are relatively future quantities at slot $t$.

Like causal scheduling, the bit allocation process is described in two stages: first the remaining bits are divided equally amongst the active slots and then bits are added/subtracted depending on the
channel state.

\section{Channel Characterization by Fractional Moments}
\label{sec:nu}

We characterize the statistics of the channel states by using the fractional moments of
the inverse of the channel states $g$.
We define the following quantity for $m=1,2,\ldots$,
\begin{equation}\label{eq:nu_m}
\nu_m = \left(\E\left[\left(\frac{1}{g}\right)^{\frac{1}{m}}\right]\right)^m.
\end{equation}
Then, the properties of these quantities are summarized as follows:
\begin{proposition} \label{prop:nu}
The channel statistics defined according to \eqref{eq:nu_m} for a non-degenerate\footnote{This eliminates the delta-type density (point-mass) function.} positive random variable have the following properties:
\begin{compactenum}[(a)]
\item the sequence $\{\nu_m\}$ is strictly decreasing and the limit exists (denote the limit as $\nu_\infty$),
\item the sequence $\{(\nu_m \nu_{m-1} \cdots \nu_1)^{1/m}\}$ is also strictly decreasing and its limit is also $\nu_\infty$.
\end{compactenum}
\end{proposition}
\begin{proof}
\begin{compactenum}[(a)]
\item First, we show the sequence $\{\nu_m\}$ is monotonically decreasing. Let $Y=1/g$ and $f_Y(y)$ be the pdf of $Y$. By the H\"older's inequality \cite{Rudin_Book87},
\begin{eqnarray*}
\E\left[Y^{\frac{1}{m+1}}\right] &=& \int_0^\infty y^{\frac{1}{m+1}} f_Y(y) dy \\
&=& \int_0^\infty \left(y^{\frac{1}{m}}f_Y(y) \right)^{\frac{m}{m+1}} \left(f_Y(y)\right)^{\frac{1}{m+1}} dy \\
&<& \left(\int_0^\infty y^{\frac{1}{m}}f_Y(y) dy\right)^{\frac{m}{m+1}} \left(\int_0^\infty f_Y(y) dy\right)^{\frac{1}{m+1}}
= \left(\E[Y^{\frac{1}{m}}]\right)^{\frac{m}{m+1}}.
\end{eqnarray*}
The strict inequality is due to the fact that $Y$ is not a point-mass density. Raising both sides to the power $(m+1)$ gives
$\nu_{m+1}< \nu_m$.

Second, we show convergence of the sequence. Let $\phi_m(y)=y^{\frac{1}{m}}$ for $y>0$ and $\psi(y)=1+y$ for $y>0$. Then, it is clear that $\lim_{m\to\infty}\phi_m(y)=1$ for all $y>0$, and $0<\phi_m(y) \le \psi(y)$ for all $y>0$. Additionally, $\int_0^\infty \psi(y) f_Y(y) dy <\infty$.
By the dominated convergence theorem \cite{Rudin_Book87}, we have
\begin{equation*}
\lim_{m\to\infty} \E [Y^{\frac{1}{m}}]=\lim_{m\to\infty}\int_0^\infty \phi_m(y) f_Y(y) dy=\int_0^\infty 1\cdot f_Y(y) dy =1.
\end{equation*}
Let $x$ be a positive real number. By the continuity of the logarithmic function, we have $\lim_{x\to 0}\ln \E[Y^x]=0$. By L'Hospital rule,
\begin{equation*}
\lim_{x\to 0} \frac{\ln\E[Y^x]}{x} = \lim_{x\to 0}\frac{\E[Y^x \ln Y]}{\E[Y^x]} = \E[\ln Y]
\end{equation*}
since $\lim_{x\to 0}\E[Y^x]=1$ and $\lim_{x\to 0}\E[Y^x \ln Y] = \E[\ln Y]$ (due to the dominated convergence theorem).
By the continuity of the exponential function, $\lim_{x\to 0}e^{\frac{1}{x}\ln \E[Y^x]}=e^{\E[\ln Y]}$. Since the above limit exists and $x$ is in the superset of integers, we have the result.

\item The monotonicity of the sequence $\{(\nu_m\nu_{m-1}\cdots\nu_1)^{1/m}\}$ follows immediately from the monotonicity of the sequence $\{\nu_m\}$ and its positivity.

By the property of the exponential function, we have $\left(\nu_m \nu_{m-1} \cdots \nu_1\right)^{\frac{1}{m}}=e^{\frac{1}{m}\ln(\nu_m \nu_{m-1} \cdots \nu_1)}=e^{\frac{1}{m}\sum_{n=1}^m \ln \nu_n}$.
Since $\lim_{m\to\infty} \nu_m = \nu_\infty$ and $\log$ is continuous, $\lim_{m\to\infty} \ln \nu_m = \ln \nu_\infty$. By Ces\'aro mean,
\begin{equation*}
\lim_{m\to\infty}\frac{1}{m}\sum_{n=1}^m \ln \nu_n = \ln \nu_\infty.
\end{equation*}
From the continuity of the exponential function, we have the result.
\end{compactenum}
\end{proof}
Notice that $\nu_1$ and $\nu_\infty$ represent the arithmetic mean and the geometric mean of random variable $1/g$, respectively. All other values in the sequence $\{\nu_m\}$ lie between these two means.

\section{Proof of Theorem \ref{thm:ratio_T2}}
\label{sec:pf_thm_ratio_T2}
For simple derivation, we work in units of nats rather than bits.
From \eqref{eq:b2_opt_T2}, the energy cost can be derived as
\begin{equation}
J_2^\text{opt}(g_2, B) = \begin{cases}
(e^B-1)\nu_1,& g_2 \le \frac{e^{-B}}{\nu_1},\\
2e^{\frac{B}{2}}\left(\frac{1}{g_2}\nu_1\right)^{1/2}-\frac{1}{g_2} -\nu_1,&  \frac{e^{-B}}{\nu_1}<g_2< \frac{e^{B}}{\nu_1}, \\
\frac{e^B-1}{g_2}, & g_2 \ge \frac{e^B}{\nu_1}.
\end{cases}
\end{equation}
Thus,
\begin{eqnarray}
\bar{J}_2^\text{opt}(B) &=& \E_{g_2}\left[J_2^\text{opt}(g_2, B)\right] \nonumber \\
&=& \int_0^{\frac{e^{-B}}{\nu_1}} (e^B-1)\nu_1 dF(x) + \int_{\frac{e^{-B}}{\nu_1}}^{\frac{e^{B}}{\nu_1}} \left[2 e^{\frac{B}{2}}\left(\frac{1}{x} \nu_1\right)^{1/2}-\frac{1}{x}-\nu_1\right] dF(x) \\
&&\qquad \qquad+\int_{\frac{e^{B}}{\nu_1}}^\infty \frac{e^B-1}{x} dF(x), \nonumber
\end{eqnarray}
where $F$ is the cumulative distribution function (CDF) of the channel state $g$.

By the limit rules, 
\begin{eqnarray}
\lim_{B\to 0} \frac{\bar{J}_2^\text{opt}(B)}{e^B-1} &=& \lim_{B\to 0} \frac{\left\{\int_0^{\frac{e^{-B}}{\nu_1}} (e^B-1)\nu_1 dF(x) +\int_{\frac{e^{B}}{\nu_1}}^\infty \frac{e^B-1}{x} dF(x) \right\}}{e^B-1} \nonumber \\
&=& \lim_{B\to 0} \frac{\int_0^\infty (e^B-1) \min\left(\frac{1}{x},\nu_1 \right) dF(x)}{e^B-1} \nonumber \\
&=& \E\left[\min\left(\frac{1}{g},\nu_1 \right)\right]
\end{eqnarray}
and
\begin{equation}
\lim_{B\to 0}\frac{e^B-1}{2(e^{B/2}-1)}=1.
\end{equation}
With \eqref{eq:EJt_eq}, we obtain \eqref{eq:J_ratio_eq_opt_B_small}.
Likewise,
\begin{eqnarray}
\lim_{B\to \infty} \frac{\bar{J}_2^\text{opt}(B)}{2e^{\frac{B}{2}}(\nu_2\nu_1)^{1/2}} &=& \lim_{B\to\infty} \frac{\int_{\frac{e^{-B}}{\nu_1}}^{\frac{e^{B}}{\nu_1}} \left[2 e^{\frac{B}{2}}\left(\frac{1}{x} \nu_1\right)^{1/2}-\frac{1}{x}-\nu_1\right] dF(x)}{2e^{\frac{B}{2}}(\nu_2\nu_1)^{1/2}} \nonumber \\
&=& \lim_{B\to\infty} \frac{2 e^{\frac{B}{2}}\left(\nu_2 \nu_1\right)^{1/2}-2\nu_1}{2e^{\frac{B}{2}}(\nu_2\nu_1)^{1/2}}=1
\end{eqnarray}
and
\begin{equation}
\lim_{B\to \infty} \frac{\bar{J}_2^\text{eq}(B)}{2e^{\frac{B}{2}}\nu_1} =1
\end{equation}
Thus, we have shown \eqref{eq:J_ratio_eq_opt_B_large}.

\section{Derivation of \eqref{eq:stopping_omega_t}}
\label{sec:pf_stopping_omega_t}
From \eqref{eq:oneshot_dp} the threshold $\omega_t$ is related to the expected cost-to-go by $\omega_t = \frac{1}{2^{B}-1}\E[J_{t-1}(B)]$. The one-step cost-to-go is $\E[J_1(B)]=(2^B-1)\E\left[\frac{1}{g}\right]$ and therefore $\omega_2=\E\left[\frac{1}{g}\right]$. For $t>2$, we expand the cost-to-go in terms of $\omega_{t-1}$ to give:
\begin{eqnarray*}
\omega_t &=& \frac{1}{2^{B}-1}\E[J_{t-1}(B)]\\
&=& \frac{1}{2^{B}-1} \Bigg(\E\left[\frac{2^{B}-1}{g_{t-1}}\Bigg\vert \frac{1}{g_{t-1}}<\omega_{t-1} \right]\Pr\left\{\frac{1}{g_{t-1}}<\omega_{t-1}\right\}+\\ && \hspace{80pt} \E\left[ \E[J_{t-2}(B)]\Bigg\vert \frac{1}{g_{t-1}}<\omega_{t-1}\right]\Pr\left\{\frac{1}{g_{t-1}}\ge\omega_{t-1}\right\}\Bigg)
\end{eqnarray*}
By substituting $\E[J_{t-2}(B)]=(2^{B}-1)\omega_{t-1}$, we have the result.

\bibliographystyle{IEEEtran}
\bibliography{scheduling}

\end{document}